\documentclass[letterpaper,english,letterpaper,twocolumn,10pt,conference]{IEEEtran}
\usepackage[T1]{fontenc}
\usepackage[latin9]{inputenc}
\usepackage{amsmath}
\usepackage{graphicx}
\usepackage{amssymb}
\usepackage{esint}

\newtheorem{thm}{Theorem}
\newtheorem{remrk}{Remark}
\newtheorem{conject}{Conjecture}

\IEEEoverridecommandlockouts
\usepackage{babel}

\begin{document}

\title{Modulation Codes for Flash Memory Based on Load-Balancing Theory}

\author{Fan Zhang and Henry D. Pfister\thanks{This work was supported in part by the National Science Foundation under Grant No. 0747470.}\\
{\normalsize Department of Electrical and Computer Engineering,
Texas A\&M University }\\
{\normalsize \{fanzhang,hpfister\}@tamu.edu}}

\maketitle
\thispagestyle{empty}\pagestyle{empty}
\begin{abstract}
In this paper, we consider modulation codes for practical multilevel
flash memory storage systems with $q$ cell levels. Instead of maximizing
the lifetime of the device \cite{Ajiang-isit07-01,Ajiang-isit07-02,Yaakobi_verdy_siegel_wolf_allerton08,Finucane_Liu_Mitzenmacher_aller08},
we maximize the average amount of information stored per cell-level,
which is defined as storage efficiency. Using this framework, we show
that the worst-case criterion \cite{Ajiang-isit07-01,Ajiang-isit07-02,Yaakobi_verdy_siegel_wolf_allerton08}
and the average-case criterion \cite{Finucane_Liu_Mitzenmacher_aller08}
are two extreme cases of our objective function. A self-randomized
modulation code is proposed which is asymptotically optimal, as $q\rightarrow\infty$,
for an arbitrary input alphabet and i.i.d. input distribution. 

In practical flash memory systems, the number of cell-levels $q$
is only moderately large. So the asymptotic performance as $q\rightarrow\infty$
may not tell the whole story. Using the tools from load-balancing
theory, we analyze the storage efficiency of the self-randomized modulation
code. The result shows that only a fraction of the cells are utilized
when the number of cell-levels $q$ is only moderately large. We also
propose a load-balancing modulation code, based on a phenomenon known
as {}``the power of two random choices'' \cite{Mitzenmacher96thepower},
to improve the storage efficiency of practical systems. Theoretical
analysis and simulation results show that our load-balancing modulation
codes can provide significant gain to practical flash memory storage
systems. Though pseudo-random, our approach achieves the same load-balancing
performance, for i.i.d. inputs, as a purely random approach based
on the power of two random choices. 
\end{abstract}

\section{Introduction\label{sec:Introduction}}

Information-theoretic research on capacity and coding for write-limited
memory originates in \cite{Rivest_Shamir_84}, \cite{Fiat_Shamir_it84},
\cite{Heegard_Gamal_it83} and \cite{wwzk_belllab_report84}. In \cite{Rivest_Shamir_84},
the authors consider a model of write-once memory (WOM). In particular,
each memory cell can be in state either 0 or 1. The state of a cell
can go from 0 to 1, but not from 1 back to 0 later. These write-once
bits are called \emph{wits}. It is shown that, the efficiency of storing
information in a WOM can be improved if one allows multiple rewrites
and designs the storage/rewrite scheme carefully. 

Multilevel flash memory is a storage technology where the charge level
of any cell can be easily increased, but is difficult to decrease.
Recent multilevel cell technology allows many charge levels to be
stored in a cell. Cells are organized into blocks that contain roughly
$10^{5}$ cells. The only way to decrease the charge level of a cell
is to erase the whole block (i.e., set the charge on all cells to
zero) and reprogram each cell. This takes time, consumes energy, and
reduces the lifetime of the memory. Therefore, it is important to
design efficient rewriting schemes that maximize the number of rewrites
between two erasures \cite{Ajiang-isit07-01}, \cite{Ajiang-isit07-02},
\cite{Yaakobi_verdy_siegel_wolf_allerton08}, \cite{Finucane_Liu_Mitzenmacher_aller08}.
The rewriting schemes increase some cell charge levels based on the
current cell state and message to be stored. In this paper, we call
a rewriting scheme a \emph{modulation code}. 

Two different objective functions for modulation codes are primarily
considered in previous work: (i) maximizing the number of rewrites
for the worst case \cite{Ajiang-isit07-01,Ajiang-isit07-02,Yaakobi_verdy_siegel_wolf_allerton08}
and (ii) maximizing for the average case \cite{Finucane_Liu_Mitzenmacher_aller08}.
As Finucane et al. \cite{Finucane_Liu_Mitzenmacher_aller08} mentioned,
the reason for considering average performance is the averaging effect
caused by the large number of erasures during the lifetime of a flash
memory device. Our analysis shows that the worst-case objective and
the average case objective are two extreme cases of our optimization
objective. We also discuss under what conditions each optimality measure
makes sense. 

In previous work (e.g., \cite{Finucane_Liu_Mitzenmacher_aller08,Ajiang-isit07-02,Ajiang-isit08-01,Yaakobi_verdy_siegel_wolf_allerton08}),
many modulation codes are shown to be asymptotically optimal as the
number of cell-levels $q$ goes to infinity. But the condition that
$q\rightarrow\infty$ can not be satisfied in practical systems. Therefore,
we also analyze asymptotically optimal modulation codes when $q$
is only moderately large using the results from load-balancing theory
\cite{load_balancing_94,Mitzenmacher96thepower,Raab-Steger-98}. This
suggests an enhanced algorithm that improves the performance of practical
system significantly. Theoretical analysis and simulation results
show that this algorithm performs better than other asymptotically
optimal algorithms when $q$ is moderately large.

The structure of the paper is as follows. The system model and performance
evaluation metrics are discussed in Section \ref{sec:Optimality-Measure}.
An asymptotically optimal modulation code, which is universal over
arbitrary i.i.d. input distributions, is proposed in Section \ref{sub:Another-rewriting-algorithm}.
The storage efficiency of this asymptotically optimal modulation code
is analyzed in Section \ref{sec:An-Enhanced-Algorithm}. An enhanced
modulation code is also presented in Section \ref{sec:An-Enhanced-Algorithm}.
The storage efficiency of the enhanced algorithm is also analyzed
in Section \ref{sec:An-Enhanced-Algorithm}. Simulation results and
comparisons are presented in Section \ref{sec:Simulation-Results}.
The paper is concluded in Section \ref{sec:Conclusion}.

\section{System Model\label{sec:Optimality-Measure}}

\subsection{System Description}

Flash memory devices usually rely on error detecting/correcting codes
to ensure a low error rate. So far, practical systems tend to use
Bose-Chaudhuri-Hocquenghem (BCH) and Reed-Solomon (RS) codes. The
error-correcting codes (ECC's) are used as the outer codes while the
modulation codes are the inner codes. In this paper, we focus on the
modulation codes and ignore the noise and the design of ECC for now. 

Let us assume that a block contains $n\times N$ $q$-level cells
and that $n$ cells (called an $n$-cell) are used together to store
$k$ $l$-ary variables (called a $k$-variable). A block contains
$N$ $n$-cells and the $N$ $k$-variables are assumed to be i.i.d.
random variables. We assume that all the $k$-variables are updated
together randomly at the same time and the new values are stored in
the corresponding $n$-cells. This is a reasonable assumption in a
system with an outer ECC. We use the subscript $t$ to denote the
time index and each rewrite increases $t$ by 1. When we discuss a
modulation code, we focus on a single $n$-cell. (The encoder of the
modulation code increases some of the cell-levels based on the current
cell-levels and the new value of the $k$-variable.) Remember that
cell-levels can only be increased during a rewrite. So, when any cell-level
must be increased beyond the maximum value $q-1$, the whole block
is erased and all the cell levels are reset to zero. We let the maximal
allowable number of block-erasures be $M$ and assume that after $M$
block erasures, the device becomes unreliable. 

Assume the $k$-variable written at time $t$ is a random variable
$x_{t}$ sampled from the set $\{0,1,\cdots,l^{k}-1\}$ with distribution
$p_{X}(x)$. For convenience, we also represent the $k$-variable
at time $t$ in the vector form as $\bar{x}_{t}\in\mathbb{Z}_{l}^{k}$
where $\mathbb{Z}_{l}$ denotes the set of integers modulo $l$. The
cell-state vector at time $t$ is denoted as $\bar{s}_{t}=(s_{t}(0),s_{t}(1),\ldots,s_{t}(n-1))$
and $s_{t}(i)\in\mathbb{Z}_{q}$ denotes the charge level of the $i$-th
cell at time $t.$ When we say $\bar{s}_{i}\succeq\bar{s}_{j},$ we
mean $s_{i}(m)\ge s_{j}(m)$ for $m=0,1,,\ldots,n-1.$ Since the charge
level of a cell can only be increased, continuous use of the memory
implies that an erasure of the whole block will be required  at some
point. Although writes, reads and erasures can all introduce noise
into the memory, we neglect this and assume that the writes, reads
and erasures are noise-free. 

Consider writing information to a flash memory when encoder knows
the previous cell state $\bar{s}_{t-1},$ the current $k$-variable
$\bar{x}_{t}$, and an encoding function $f:\mathbb{Z}_{l}^{k}\times\mathbb{Z}_{q}^{n}\rightarrow\mathbb{Z}_{q}^{n}$
that maps $\bar{x}_{t}$ and $\bar{s}_{t-1}$ to a new cell-state
vector $\bar{s}_{t}$. The decoder only knows the current cell state
$\bar{s}_{t}$ and the decoding function $g:\mathbb{Z}_{q}^{n}\rightarrow\mathbb{Z}_{l}^{k}$
that maps the cell state $\bar{s}_{t}$ back to the variable vector
$\bar{\hat{x}}_{t}$. Of course, the encoding and decoding functions
could change over time to improve performance, but we only consider
time-invariant encoding/decoding functions for simplicity.

\subsection{Performance Metrics}

\subsubsection{Lifetime v.s. Storage Efficiency}

The idea of designing efficient modulation codes jointly to store
multiple variables in multiple cells was introduced by Jiang \cite{Ajiang-isit07-01}.
In previous work on modulation codes design for flash memory (e.g.
\cite{Ajiang-isit07-01}, \cite{Ajiang-isit07-02}, \cite{Yaakobi_verdy_siegel_wolf_allerton08},
\cite{Finucane_Liu_Mitzenmacher_aller08}), the lifetime of the memory
(either worst-case or average) is maximized given fixed amount of
information per rewrite. Improving storage density and extending the
lifetime of the device are two conflicting objectives. One can either
fix one and optimize the other or optimize over these two jointly.
Most previous work (e.g., \cite{Finucane_Liu_Mitzenmacher_aller08,Ajiang-isit07-02,Ajiang-isit08-01,Yaakobi_verdy_siegel_wolf_allerton08})
takes the first approach by fixing the amount of information for each
rewrite and maximizing the number of rewrites between two erasures.
In this paper, we consider the latter approach and our objective is
to maximize the total amount of information stored in the device until
the device dies. This is equivalent to maximizing the average (over
the $k$-variable distribution $p_{X}(x)$) amount of information
stored per cell-level, \[
\gamma\triangleq E\left(\frac{\sum_{i=1}^{R}I_{i}}{n(q-1)}\right),\]
 where $I_{i}$ is the amount of information stored at the $i$-th
rewrite, $R$ is the number of rewrites between two erasures, and
the expectation is over the $k$-variable distribution. We also call
$\gamma$ as \emph{storage efficiency}.

\subsubsection{Worst Case v.s. Average Case}

In previous work on modulation codes for flash memory, the number
of rewrites of an $n$-cell has been maximized in two different ways.
The authors in \cite{Ajiang-isit07-01,Ajiang-isit07-02,Yaakobi_verdy_siegel_wolf_allerton08}
consider the worst case number of rewrites and the authors in \cite{Finucane_Liu_Mitzenmacher_aller08}
consider the average number of rewrites. As mentioned in \cite{Finucane_Liu_Mitzenmacher_aller08},
the reason for considering the average case is due to the large number
of erasures in the lifetime of a flash memory device. Interestingly,
these two considerations can be seen as two extreme cases of the optimization
objective in (\ref{eq:opt}).

Let the $k$-variables be a sequence of i.i.d. random variables over
time and all the $n$-cells. The objective of optimization is to maximize
the amount of information stored until the device dies. The total
amount of information stored in the device\footnote{There is a subtlety here. If the $n$-cell changes to the same value, should it count as stored information? Should this count as a rewrite? This formula assumes that it counts as a rewrite, so that $l^k$ values (rather than $l^k -1$) can be stored during each rewrite.}
can be upper-bounded by \begin{equation}
W=\sum_{i=1}^{M}R_{i}\log_{2}(l^{k})\label{eq:total_info_ub}\end{equation}
where $R_{i}$ is the number of rewrites between the $(i-1)$-th and
the $i$-th erasures. Note that the upper bound in  (\ref{eq:total_info_ub})
is achievable by uniform input distribution, i.e., when the input
$k$-variable is uniformly distributed over $\mathbb{Z}_{l^{k}}$,
each rewrite stores $\log_{2}(l^{k})=k\log_{2}l$ bits of information.
Due to the i.i.d. property of the input variables over time, $R_{i}$'s
are i.i.d. random variables over time. Since $R_{i}$'s are i.i.d.
over time, we can drop the subscript $i$. Since $M$, which is the
maximum number of erasures allowed, is approximately on the order
of $10^{7}$, by the law of large numbers (LLN), we have \[
W\approx ME\left[R\right]k\log_{2}(l).\]
Let the set of all valid encoder/decoder pairs be \[
\mathcal{Q}=\left\{ f,g|\bar{s}_{t}=f(\bar{s}_{t-1},\bar{x}_{t}),\bar{x}_{t}=g(\bar{s}_{t}),\bar{s}_{t}\succeq\bar{s}_{t-1}\right\} ,\]
where $\bar{s}_{t}\succeq\bar{s}_{t-1}$ implies the charge levels
are element-wise non-decreasing. This allows us to treat the problem
\[
\max_{f,g\in\mathcal{Q}}W,\]
as the following equivalent problem \begin{equation}
\max_{f,g\in\mathcal{Q}}E\left[R\right]k\log_{2}(l).\label{eq:opt2-1}\end{equation}

Denote the maximal charge level of the $i$-th $n$-cell at time $t$
as $d_{i}(t)$. Note that time index $t$ is reset to zero when a
block erasure occurs and increased by one at each rewrite otherwise.
Denote the maximal charge level in a block at time $t$ as $d(t),$
which can be calculated as $d(t)=\max_{i}d_{i}(t).$ Define $t_{i}$
as the time when the $i$-th $n$-cell reaches its maximal allowed
value, i.e., $t_{i}\triangleq\min\{t|d_{i}(t)=q\}$. We assume, perhaps
naively, that a block-erasure is required when any cell within a block
reaches its maximum allowed value. The time when a block erasure is
required is defined as $T\triangleq\min_{i}t_{i}.$ It is easy to
see that $E\left[R\right]=NE\left[T\right],$ where the expectations
are over the $k$-variable distribution. So maximizing $E\left[T\right]$
is equivalent to maximizing $W$. So the optimization problem (\ref{eq:opt2-1})
can be written as the following optimization problem \begin{equation}
\max_{f,g\in\mathcal{Q}}E\left[\min_{i\in\{1,2,\cdots,N\}}t_{i}\right].\label{eq:opt3}\end{equation}
Under the assumption that the input is i.i.d. over all the $n$-cells
and time indices, one finds that the $t_{i}$'s are i.i.d. random
variables. Let their common probability density function (pdf) be
$f_{t}(x).$ It is easy to see that $T$ is the minimum of $N$ i.i.d.
random variables with pdf $f_{t}(x).$ Therefore, we have $f_{T}(x)=Nf_{t}(x)\left(1-F_{t}(x)\right)^{N-1},$
where $F_{t}(x)$ is the cumulative distribution function (cdf) of
$t_{i}.$ So, the optimization problem (\ref{eq:opt3}) becomes \begin{equation}
\max_{f,g\in\mathcal{Q}}E\left[T\right]=\max_{f,g\in\mathcal{Q}}\int Nf_{t}(x)\left(1-F_{t}(x)\right)^{N-1}x\mbox{d}x.\label{eq:opt}\end{equation}
Note that when $N=1,$ the optimization problem in  (\ref{eq:opt})
simplifies to \begin{equation}
\max_{f,g\in\mathcal{Q}}E\left[t_{i}\right].\label{eq:opt2}\end{equation}
 This is essentially the case that the authors in \cite{Finucane_Liu_Mitzenmacher_aller08}
consider. When the whole block is used as one $n$-cell and the number
of erasures allowed is large, optimizing the average (over all input
sequences) number of rewrites of an $n$-cell is equivalent to maximizing
the total amount of information stored $W.$ The analysis also shows
that the reason we consider average performance is not only due to
the averaging effect caused by the large number of erasures. One other
important assumption is that there is only one $n$-cell per block. 

The other extreme is when $N\gg1.$ In this case, the pdf $Nf_{t}(x)\left(1-F_{t}(x)\right)^{N-1}$
tends to a point mass at the minimum of $t$ and the integral $\int Nf_{t}(x)\left(1-F_{t}(x)\right)^{N-1}t\mbox{d}t$
approaches the minimum of $t$. This gives the worst case stopping
time for the programming process of an $n$-cell. This case is considered
by \cite{Ajiang-isit07-01,Ajiang-isit07-02,Yaakobi_verdy_siegel_wolf_allerton08}.
Our analysis shows that we should consider the worst case when $N\gg1$
even though the device experiences a large number of erasures. So
the optimality measure is not determined only by $M$, but also by
$N.$ When $N$ and $M$ are large, it makes more sense to consider
the worst case performance. When $N=1$, it is better to consider
the average performance. When $N$ is moderately large, we should
maximize the number of rewrites using (\ref{eq:opt}) which balances
the worst case and the average case. 

When $N$ is moderately large, one should probably focus on optimizing
the function in (\ref{eq:opt}), but it is not clear how to do this
directly. So, this remains an open problem for future research. Instead,
we will consider a load-balancing approach to improve practical systems
where $q$ is moderately large.

\subsection{$N=1$ v.s. $N\gg1$}

If we assume that there is only one variable changed each time, the
average amount of information per cell-level can be bounded by $\log_{2}kl$
because there are $kl$ possible new values. Since the number of rewrites
can be bounded by $n(q-1),$ we have \begin{equation}
\gamma\leq\log_{2}kl.\label{eq:storage_efficiency_bound2}\end{equation}
 If we allow arbitrary change on the $k$-variables, there are totally
$l^{k}$ possible new values. It can be shown that \begin{equation}
\gamma\leq k\log_{2}l.\label{eq:storage_efficiency_bound}\end{equation}
For fixed $l$ and $q$, the bound in  (\ref{eq:storage_efficiency_bound})
suggests using a large $k$ can improve the storage efficiency. This
is also the reason jointly coding over multiple cells can improve
the storage efficiency \cite{Ajiang-isit07-01}. Since optimal rewriting
schemes only allow a single cell-level to increase by one during each
rewrite, decodability implies that $n\ge kl-1$ for the first case
and $n\ge l^{k}-1$ for the second case. Therefore, the bounds in
 (\ref{eq:storage_efficiency_bound2}) and  (\ref{eq:storage_efficiency_bound})
also require large $n$ to improve storage efficiency. 

The upper bound in  (\ref{eq:storage_efficiency_bound}) grows linearly
with $k$ while the upper bound in  (\ref{eq:storage_efficiency_bound2})
grows logarithmically with $k$. Therefore, in the remainder of this
paper,  we assume an arbitrary change in the $k$-variable per rewrite
and $N=1$, i.e., the whole block is used as an $n$-cell, to improve
the storage efficiency. This approach implicitly trades instantaneous
capacity for future storage capacity because more cells are used to
store the same number of bits, but the cells can also be reused many
more times.

Note that the assumption of $N=1$ might be difficult for real implementation,
but its analysis gives an upper bound on the storage efficiency. From
the analysis above with $N=1$, we also know that maximizing $\gamma$
is equivalent to maximize the average number of rewrites.

\section{\label{sub:Another-rewriting-algorithm}Self-randomized Modulation
Codes\label{sec:Another-Rewriting-Algorithm}}

In \cite{Finucane_Liu_Mitzenmacher_aller08}, modulation codes are
proposed that are asymptotically optimal (as $q$ goes to infinity)
in the average sense when $k=2$. In this section, we introduce a
modulation code that is asymptotically optimal for arbitrary input
distributions and arbitrary $k$ and $l$. This rewriting algorithm
can be seen as an extension of the one in \cite{Finucane_Liu_Mitzenmacher_aller08}.
The goal is, to increase the cell-levels uniformly on average for
an arbitrary input distribution. Of course, decodability must be maintained.
The solution is to use common information, known to both the encoder
(to encode the input value) and the decoder (to ensure the decodability),
to randomize the cell index over time for each particular input value.

Let us assume the $k$-variable is an i.i.d. random variable over
time with arbitrary distribution $p_{X}(x)$ and the $k$-variable
at time $t$ is denoted as $x_{t}\in\mathbb{Z}_{l^{k}}.$ The output
of the decoder is denoted as $\hat{x}_{t}\in\mathbb{Z}_{l^{k}}.$
We choose $n=l^{k}$ and let the cell state vector at time $t$ be
$\bar{s}_{t}=(s_{t}(0),s_{t}(1),\cdots,s_{t}(n-1))$, where $s_{t}(i)\in\mathbb{Z}_{q}$
is the charge level of the $i$-th cell at time $t.$ At $t=0$, the
variables are initialized to $\overline{s}_{0}=(0,\ldots,0)$, $x_{0}=0$
and $r_{0}=0$.

The decoding algorithm $\hat{x}_{t}=g(\bar{s}_{t})$ is described
as follows.
\begin{itemize}
\item Step 1: Read cell state vector $\bar{s}_{t}$ and calculate the $\ell_{1}$
norm $r_{t}=\|\bar{s}_{t}\|_{1}$.
\item Step 2: Calculate $s_{t}=\sum_{i=1}^{n-1}is_{t}(i)$ and $\hat{x}_{t}=s_{t}-\frac{r_{t}(r_{t}+1)}{2}\bmod l^{k}.$
\end{itemize}
The encoding algorithm $\overline{s}_{t}=f(\overline{s}_{t-1},x_{t})$
is described as follows.
\begin{itemize}
\item Step 1: Read cell state $\bar{s}_{t-1}$ and calculate $r_{t-1}$
and $\hat{x}_{t-1}$ as above. If $\hat{x}_{t-1}=x_{t,}$ then do
nothing.
\item Step 2: Calculate $\Delta x_{t}=x_{t}-\hat{x}_{t-1}\bmod l^{k}$ and
$w_{t}=\Delta x_{t}+r_{t-1}+1\bmod l^{k}$ 
\item Step 3: Increase the charge level of the $w_{t}$-th cell by 1.
\end{itemize}
For convenience, in the rest of the paper, we refer the above rewriting
algorithm as {}``self-randomized modulation code''. 
\begin{thm}
The self-randomized modulation code achieves at least $n(q-q^{2/3})$
rewrites with high probability, as $q\rightarrow\infty,$ for arbitrary
$k,$ $l,$ and i.i.d. input distribution $p_{X}(x)$. Therefore,
it is asymptotically optimal for random inputs as $q\rightarrow\infty$. \end{thm}
\begin{proof}
[Sketch of Proof] The proof is similar to the proof in \cite{Finucane_Liu_Mitzenmacher_aller08}.
Since exactly one cell has its level increased by 1 during each rewrite,
$r_{t}$ is an integer sequence that increases by 1 at each rewrite.
The cell index to be written $w_{t}$ is randomized by adding the
value $(r_{t}+1)\bmod l^{k}$. This causes each consecutive sequence
of $l^{k}$ rewrites to have a uniform affect on all cell levels.
As $q\rightarrow\infty$, an unbounded number of rewrites is possible
and we can assume $t\rightarrow\infty$. 

Consider the first $nq-nq^{2/3}$ steps, the value $a_{t,k,l}\triangleq(r_{t}+1)\mbox{ mod }l^{k}$
is as even as possible over $\{0,1,\cdots,l^{k}-1\}.$ For convenience,
we say there are $(q-q^{2/3})$ $a_{t,k,l}$'s at each value, as the
rounding difference by 1 is absorbed in the $o(q)$ term. Assuming
the input distribution is $p_{X}=\{p_{0},p_{1},\cdots,p_{l^{k}-1}\}$.
For the case that $a_{t,k,l}=i$, the probability that $w_{t}=j$
is $p_{(j-i)\mbox{ mod }l^{k}}$ for $j\in\{0,1,\cdots,l^{k}-1\}$.
Therefore, $w_{j}$ has a uniform distribution over $\{0,1,\cdots,l^{k}-1\}$.
Since inputs are independent over time, by applying the same Chernoff
bound argument as \cite{Finucane_Liu_Mitzenmacher_aller08}, it follows
that the number of times $w_{t}=j$ is at most $q-3$ with high probability
(larger than $1-1/\mbox{poly}(q)$) for all $j$. Summing over $j$,
we finish the proof. \end{proof}
\begin{remrk}
Notice that the randomizing term $r_{t}$ a deterministic term which
makes $w_{t}$ look  \emph{random} over time in the sense that there
are equally many terms for each value. Moreover, $r_{t}$ is known
to both the encoder and the decoder such that the encoder can generate
{}``uniform'' cell indices over time and the decoder knows the accumulated
value of $r_{t}$, it can subtract it out and recover the data correctly.
Although this algorithm is asymptotically optimal as $q\rightarrow\infty$,
the maximum number of rewrites $n(q-o(q))$ cannot be achieved for
moderate $q$. This motivates the analysis and the design of an enhanced
version of this algorithm for practical systems in next section. 
\end{remrk}

\begin{remrk}
A self-randomized modulation code uses $n=l^{k}$ cells to store a
$k$-variable. This is much larger than the $n=kl$ used by previous
asymptotically optimal algorithms because we allow the $k$-variable
to change arbitrarily. Although this seems to be a waste of cells,
the average amount of information stored per cell-level is actually
maximized (see (\ref{eq:storage_efficiency_bound2}) and (\ref{eq:storage_efficiency_bound})).
In fact, the definition of asymptotic optimality requires  $n\ge l^{k}-1$
if we allow arbitrary changes to the $k$-variable. 
\end{remrk}

\begin{remrk}
We note that the optimality of the self-randomized modulation codes
is similar to the weak robust codes presented in \cite{Ajiang-isit09-01}. 
\end{remrk}

\begin{remrk}
We use $n=l^{k}$ cells to store one of $l^{k}-1$ possible messages.
This is slightly worse than the simple method of using $n=l^{k}-1$.
Is it possible to have self-randomization using only $n=l^{k}-1$
cells? A preliminary analysis of this question based on group theory
indicates that it is not. Thus, the extra cell provides the possibility
to randomize the mappings between message values and the cell indices
over time. 
\end{remrk}

\section{Load-balancing Modulation Codes \label{sec:An-Enhanced-Algorithm}}

While asymptotically optimal modulation codes (e.g., codes in \cite{Ajiang-isit07-01},
\cite{Ajiang-isit07-02}, \cite{Yaakobi_verdy_siegel_wolf_allerton08},
\cite{Finucane_Liu_Mitzenmacher_aller08} and the self-randomized
modulation codes described in Section \ref{sec:Another-Rewriting-Algorithm})
require $q\rightarrow\infty$, practical systems use $q$ values between
$2$ and $256$. Compared to the number of cells $n$, the size of
$q$ is not quite large enough for asymptotic optimality to suffice.
In other words, codes that are asymptotically optimal may have significantly
suboptimal performance when the system parameters are not large enough.
Moreover, different asymptotically optimal codes may perform differently
when $q$ is not large enough. Therefore, asymptotic optimality can
be misleading in this case. In this section, we first analyze the
storage efficiency of self-randomized modulation codes when $q$ is
not large enough and then propose an enhanced algorithm which improves
the storage efficiency significantly.

\subsection{Analysis for Moderately Large $q$}

Before we analyze the storage efficiency of asymptotically optimal
modulation codes for moderately large $q$, we first show the connection
between rewriting process and the load-balancing problem (aka the
balls-into-bins or balls-and-bins problem) which is well studied in
mathematics and computer science \cite{load_balancing_94,Mitzenmacher96thepower,Raab-Steger-98}.
Basically, the load-balancing problem considers how to distribute
objects among a set of locations as evenly as possible. Specifically,
the balls-and-bins model considers the following problem. If $m$
balls are thrown into $n$ bins, with each ball being placed into
a bin chosen independently and uniformly at random, define the \emph{load}
as the number of balls in a bin, what is the maximal load over all
the bins? Based on the results in Theorem 1 in \cite{Raab-Steger-98},
we take a simpler and less accurate approach to the balls-into-bins
problem and arrive at the following theorem. 
\begin{thm}
\label{thm:random_loading}Suppose that $m$ balls are sequentially
placed into $n$ bins. Each time a bin is chosen independently and
uniformly at random. The maximal load over all the bins is $L$ and:

($i$) If $m=d_{1}n,$ the maximally loaded bin has $L\le\frac{c_{1}\ln n}{\ln\ln n}$
balls, $c_{1}>2$ and $d_{1}\ge1$, with high probability ($1-1/\mbox{poly}(n)$)
as $n\rightarrow\infty.$

($ii$) If $m=n\ln n$, the maximally loaded bin has $L\le\frac{c_{4}(\ln n)^{2}}{\ln\ln n}$
balls, $c_{4}>1$, with high probability ($1-1/\mbox{poly}(n)$) as
$n\rightarrow\infty.$

($iii$) If $m=c_{3}n^{d_{2}},$ the maximally loaded bin has $L\le ec_{3}n^{d_{2}-1}+c_{2}\ln n$,
$c_{2}>1$, $c_{3}\ge1$ and $d_{2}>1$, with high probability ($1-1/\mbox{poly}(n)$)
as $n\rightarrow\infty.$ \end{thm}
\begin{proof}
Denote the event that there are at least $k$ balls in a particular
bin as $E_{k}$. Using the union bound over all subsets of size $k,$
it is easy to show that the probability that $E_{k}$ occurs is upper
bounded by \[
\Pr\{E_{k}\}\le\binom{m}{k}\left(\frac{1}{n}\right)^{k}.\]
 Using Stirling's formula, we have $\binom{m}{k}\le\left(\frac{me}{k}\right)^{k}$.
Then $\Pr\{E_{k}\}$ can be further bounded by \begin{equation}
\Pr\{E_{k}\}\le\left(\frac{me}{nk}\right)^{k}.\label{eq:maxload_ub}\end{equation}
 If $m=d_{1}n$, substitute $k=\frac{c_{1}\ln n}{\ln\ln n}$ to the
RHS of  (\ref{eq:maxload_ub}), we have \begin{align*}
\Pr\{E_{k}\} & \le\left(\frac{d_{1}e\ln\ln n}{c_{1}\ln n}\right)^{\frac{c_{1}\ln n}{\ln\ln n}}\\
 & =e^{\left(\frac{c_{1}\ln n}{\ln\ln n}\left(\ln(d_{1}e\ln\ln n)-\ln(c_{1}\ln n)\right)\right)}\\
 & <e^{\left(\frac{c_{1}\ln n}{\ln\ln n}\left(\ln(d_{1}e\ln\ln n)-\ln\ln n\right)\right)}\\
 & \le e^{(-(c_{1}-1)\ln n)}=\frac{1}{n^{c_{1}-1}}.\end{align*}
Denote the event that all bins have at most $k$ balls as $\tilde{E}_{k}$.
By applying the union bound, it is shown that \[
\Pr\{\tilde{E}_{k}\}\ge1-\frac{n}{n^{c_{1}-1}}=1-\frac{1}{n^{c_{1}-2}}.\]
 Since $c_{1}>2,$ we finish the proof for the case of $m=d_{1}n.$ 

If $m=n\ln n$, substitute $k=\frac{c_{4}(\ln n)^{2}}{\ln\ln n}$
to the RHS of  (\ref{eq:maxload_ub}), we have \begin{align*}
\Pr\{E_{k}\} & \le\left(\frac{e\ln\ln n}{c_{4}\ln n}\right)^{\frac{c_{4}(\ln n)^{2}}{\ln\ln n}}\\
 & =e^{\left(\frac{c_{4}(\ln n)^{2}}{\ln\ln n}\left(\ln e\ln\ln n-\ln c_{4}\ln n\right)\right)}\\
 & \le e^{\left(\frac{c_{4}(\ln n)^{2}}{\ln\ln n}\left(\ln e\ln\ln n-\ln\ln n\right)\right)}\\
 & \le e^{\left(-(c_{4}-1)(\ln n)^{2}\right)}=o\left(\frac{1}{n^{2}}\right).\end{align*}
By applying the union bound, we finish the proof for the case of $m=n\ln n.$ 

If $m=c_{3}n^{d_{2}},$ substitute $k=ec_{3}n^{d_{2}-1}+c_{2}\ln n=ec_{3}n^{d_{2}-1}+c_{2}\ln n$
to the RHS of (\ref{eq:maxload_ub}), we have \begin{align*}
\Pr\{E_{k}\} & \le\left(\frac{ec_{3}n^{d_{2}-1}}{ec_{3}n^{d_{2}-1}+c_{2}\ln n}\right)^{c_{3}en^{d_{2}-1}+c_{2}\ln n}\\
 & =e^{\left((c_{5}n^{d_{2}-1}\!+c_{2}\ln n)\left(\ln c_{5}n^{d_{2}-1}\!\!\!\!-\ln(c_{5}n^{d_{2}-1}\!+c_{2}\ln n)\right)\right)}\\
 & \le e^{\left((c_{5}n^{d_{2}-1}+c_{2}\ln n)\left(-\frac{c_{2}\ln n}{c_{5}n^{d_{2}-1}}\right)\right)}\\
 & \le e^{\left(-c_{2}\ln n\right)}=\frac{1}{n^{c_{2}}}\end{align*}
 where $c_{5}=ec_{3}.$ By applying the union bound, it is shown that
\[
\Pr\{\tilde{E}_{k}\}\le1-\frac{n}{n^{c_{2}}}=1-\frac{1}{n^{c_{2}-1}}.\]
Since $c_{2}>1,$ we finish the proof for the case of $m=c_{3}n^{d_{2}}.$ \end{proof}
\begin{remrk}
Note that Theorem \ref{thm:random_loading} only shows an upper bound
on the maximum load $L$ with a simple proof. More precise results
can be found in Theorem 1 of \cite{Raab-Steger-98}, where the exact
order of $L$ is given for different cases. It is worth mentioning
that the results in Theorem 1 of \cite{Raab-Steger-98} are  different
from Theorem \ref{thm:random_loading} because Theorem 1 of \cite{Raab-Steger-98}
holds with probability $1-o(1)$ while Theorem \ref{thm:random_loading}
holds with probability ($1-1/\mbox{poly}(n)$). 
\end{remrk}

\begin{remrk}
The asymptotic optimality in the rewriting process implies that each
rewrite only increases the cell-level of a cell by 1 and all the cell-levels
are fully used when an erasure occurs. This actually implies $\lim_{m\rightarrow\infty}\frac{L}{m/n}=1$.
Since $n$ is usually a large number and $q$ is not large enough
in practice, the theorem shows that, when $q$ is not large enough,
asymptotic optimality is not achievable. For example, in practical
systems, the number of cell-levels $q$ does not depend on the number
of cells in a block. Therefore, rather than $n(q-1),$ only roughly
$n(q-1)\frac{\ln\ln n}{\ln n}$ charge levels can be used as $n\rightarrow\infty$
if $q$ is a small constant which is independent of $n$. In practice,
this loss could be mitigated by using writes that increase the charge
level in multiple cells simultaneously (instead of erasing the block).\end{remrk}
\begin{thm}
\label{thm:gamma1}The self-randomized modulation code has storage
efficiency $\gamma=c\ln\frac{k\ln l}{c}$ when $q-1=c$ and $\gamma=\frac{c}{\theta}k\ln l$
when $q-1=c\ln n$ as $n$ goes to infinity with high probability
(i.e., $1-o(1)$).\end{thm}
\begin{proof}
Consider the problem of throwing $m$ balls into $n$ bins and let
the r.v. $M$ be the number of balls thrown into $n$ bins until some
bin has more than $q-1$ balls in it. While we would like to calculate
$E[M]$ exactly, we still settle for an approximation based on the
following result. If $m=cn\ln n$, then there is a constant $d(c)$
such that maximum number of balls $L$ in any bin satisfies\[
\left(d(c)-1\right)\ln n\leq L\leq d(c)\ln n\]
with probability $1-o(1)$ as $n\rightarrow\infty$ \cite{Raab-Steger-98}
. The constant $d(c)$ is given by the largest $x$-root of\[
x(\ln c-\ln x+1)+1-c=0,\]
and solving this equation for $c$ gives the implicit expression $c=-d(c)W\left(-e^{-1-\frac{1}{d(c)}}\right)$.
Since the lower bound matches the expected maximum value better, we
define $\theta\triangleq d(c)-1$ and apply it to our problem using
the equation $\theta\ln n=q-1$ or $\theta=\frac{q-1}{\ln n}$. Therefore,
the storage efficiency is $\gamma=\frac{m\ln n}{n(q-1)}=\frac{c}{\theta}k\ln l$ 

If $m=cn$, the maximum load is approximately $\frac{\ln n}{\ln\frac{n\ln n}{m}}$
with probability $1-o(1)$ for large $n$ \cite{Raab-Steger-98}.
By definition, $q-1=\frac{\ln n}{\ln\frac{n\ln n}{m}}$. Therefore,
the storage efficiency is $\gamma=\frac{m\ln n}{n(q-1)}=c\ln\frac{\ln n}{c}=c\ln\frac{k\ln l}{c}.$ \end{proof}
\begin{remrk}
The results in Theorem \ref{thm:gamma1} show that when $q$ is on
the order of $O(\ln n)$, the storage efficiency is on the order of
$\Theta(k\ln l)$. Taking the limit as $q,n\rightarrow\infty$ with
$q=O(\ln n)$, we have $\lim\frac{\gamma}{k\ln l}=\frac{\theta}{c}>0.$
When $q$ is a constant independent of $n$, the storage efficiency
is on the order of $\Theta(\ln k\ln l).$ Taking the limit as $n\rightarrow\infty$
with $q-1=c$, we have $\lim\frac{\gamma}{k\ln l}=0$. In this regime,
the self-randomized modulation codes actually perform very poorly
even though they are asymptotically optimal as $q\rightarrow\infty$.
\end{remrk}

\subsection{\label{sub:An-enhanced-algorithm}Load-balancing Modulation Codes}

Considering the bins-and-balls problem, can we distribute balls more
evenly when $m/n$ is on the order of $o(n)?$ Fortunately, when $m=n$,
the maximal load can be reduced by a factor of roughly $\frac{\ln n}{(\ln\ln n)^{2}}$
by using \emph{the power of two random choices} \cite{Mitzenmacher96thepower}.
In detail, the strategy is, every time we pick two bins independently
and uniformly at random and throw a ball into the less loaded bin.
By doing this, the maximally loaded bin has roughly $\frac{\ln\ln n}{\ln2}+O(1)$
balls with high probability. Theorem 1 in \cite{load_balancing_94}
gives the answer in a general form when we consider $d$ random choices.
The theorem shows there is a large gain when the number of random
choice is increased from 1 to 2. Beyond that, the gain is on the same
order and only the constant can be improved. 

Based on the idea of 2 random choices, we define the following load-balanced
modulation code. 

Again, we let the cell state vector at time $t$ be $\bar{s}_{t}=(s_{t}(0),s_{t}(1),\cdots,s_{t}(n-1))$,
where $s_{t}(i)\in\mathbb{Z}_{q}$ is the charge level of the $i$-th
cell at time $t.$ This time, we use $n=l^{k+1}$ cells to store a
$k$-variable $x_{t}\in\mathbb{Z}_{l^{k}}$ (i.e., we write $(k+1)\log_{2}l$
bits to store $k\log_{2}l$ bits of information). The information
loss provides $l$ ways to write the same value. This flexibility
allows us to avoid sequences of writes that increase one cell level
too much. We are primarily interested in binary variables with 2 random
choices or $l=2$. For the power of $l$ choices to be effective,
we must try to randomize (over time), the $l$ possible choices over
the set of all $\binom{n}{l}$ possibilities. The value $r_{t}=\|\bar{s}_{t}\|_{1}$
is used to do this. Let $H$ be the Galois field with $l^{k+1}$ elements
and $h:\mathbb{Z}_{l^{k+1}}\rightarrow H$ be a bijection that satisfies
$h(0)=0$ (i.e., the Galois field element 0 is associated with the
integer 0).

The decoding algorithm calculates $\hat{x}_{t}$ from $\bar{s}_{t}$
and operates as follows:
\begin{itemize}
\item Step 1: Read cell state vector $\bar{s}_{t}$ and calculate the $\ell_{1}$
norm $r_{t}=\|\bar{s}_{t}\|_{1}$.
\item Step 2: Calculate $s_{t}=\sum_{i=1}^{n}is_{t}(i)$ and $\hat{x}_{t}'=s_{t}\mbox{ mod }l^{k+1}.$
\item Step 3: Calculate $a_{t}=h\left(\left(r_{t}\bmod l^{k}-1\right)+1\right)$
and $b_{t}=h\left(r_{t}\bmod l^{k}\right)$
\item Step 4: Calculate $\hat{x}_{t}=h^{-1}\left(a_{t}^{-1}\left(h(\hat{x}_{t}')-b_{t}\right)\right)\bmod l^{k}$.
\end{itemize}
The encoding algorithm stores $x_{t}$ and operates as follows.
\begin{itemize}
\item Step 1: Read cell state $\bar{s}_{t-1}$ and decode to $\hat{x}_{t-1}'$
and $\hat{x}_{t-1}$. If $\hat{x}_{t-1}=x_{t},$ then do nothing.
\item Step 2: Calculate $r_{t}=\|\bar{s}_{t-1}\|_{1}+1$, $a_{t}=h\left(\left(r_{t}\bmod l^{k}-1\right)+1\right)$,
and $b_{t}=h\left(r_{t}\bmod l^{k}\right)$
\item Step 3: Calculate $x_{t}^{(i)}=h^{-1}\left(a_{t}h(x_{t}+il^{k})+b_{t}\right)$
and $\Delta x_{t}^{(i)}=x_{t}^{(i)}-\hat{x}_{t-1}'\mbox{ mod }l^{k+1}$
for $i=0,1,\ldots l-1$.
\item Step 4: Calculate%
\footnote{Ties can be broken arbitrarily.%
} $w_{t}=\arg\min_{j\in\mathbb{Z}_{l}}\{s_{t-1}(\Delta x_{t}^{(j)})\}$.
Increase the charge level by 1 of cell $\Delta x_{t}^{(w_{t})}$.
\end{itemize}
Note that the state vector at $t=0$ is initialized to $s_{0}=(0,\ldots,0)$
and therefore $x_{0}=0$. The first arbitrary value that can be stored
is $x_{1}$. 

The following conjecture suggests that the ball-loading performance
of the above algorithm is identical to the random loading algorithm
with $l=2$ random choices. 
\begin{conject}
\label{thm:gamma2}If $l=2$ and $q-1=c\ln n$, then the load-balancing
modulation code has storage efficiency $\gamma=k$ with probability
1-$o(1)$ as $n\rightarrow\infty$. If $q-1=c,$ the storage efficiency
$\gamma=\frac{c\ln2}{\ln\ln n}k$ with probability 1-$o(1)$.\end{conject}
\begin{proof}
[Sketch of Proof] Consider the affine permutation $\pi_{x}^{(a,b)}=h^{-1}(ah(x)+b)$
for $a\in H\backslash0$ and $b\in H$. As $a,b$ vary, this permutation
maps the two elements $x_{t}$ and $x_{t}+l^{k}$ uniformly over all
pairs of cell indices. After $m=n(n-1)$ steps, we see that all pairs
of $a,b$ occur equally often. Therefore, by picking the less charged
cell, the modulation code is almost identical to the random loading
algorithm with two random choices. Unfortunately, we are interested
in the case where $m\ll n^{2}$ so the analysis is somewhat more delicate.
If $m=cn\ln n$, the highest charge level is $c\ln n-1+\frac{\ln\ln n}{\ln2}\approx c\ln n$
with probability $1-o(1)$ \cite{load_balancing_94}. Since $q-1=c\ln n$
in this case, the storage efficiency is $\gamma=\frac{cn\ln n\log_{2}2^{k}}{nc\ln n}=k$.
If $m=cn$, then $q-1=c$ and the maximum load is $c-1+\ln\ln n/\ln2\approx\frac{\ln\ln n}{\ln2}$.
By definition, we have $\frac{\ln\ln n}{\ln2}=q-1.$ Therefore, we
have $\gamma=\frac{cn\log_{2}2^{k}}{n(q-1)}=\frac{c\ln2}{\ln\ln n}k.$
\end{proof}

\begin{remrk}
If $l=2$ and $q$ is on the order of $O(\ln n),$ Conjecture \ref{thm:gamma2}
shows that the bound (\ref{eq:storage_efficiency_bound}) is achievable
by load-balancing modulation codes as $n$ goes to infinity. In this
regime, the load-balancing modulation codes provide a better constant
than self-randomized modulation codes by using twice many cells.
\end{remrk}

\begin{remrk}
\label{rem:If--is}If $l=2$ and $q$ is a constant independent of
$n$, the storage efficiency is $\gamma_{1}=c\ln\frac{k}{c}$ for
the self-randomized modulation code and $\gamma_{2}=\frac{c\ln2}{\ln\ln n}\log_{2}\frac{n}{2}$
for the load-balancing modulation code. But, the self-randomized modulation
code uses $n=2^{k}$ cells and the load-balancing modulation code
uses $n=2^{k+1}$ cells. To make fair comparison on the storage efficiency
between them, we let $n=2^{k+1}$ for both codes. Then we have $\gamma_{1}=c\ln\frac{\log_{2}n}{c}$
and $\gamma_{2}=\frac{c\ln2}{\ln\ln n}\log_{2}\frac{n}{2}$. So, as
$n\rightarrow\infty$, we see that $\frac{\gamma_{1}}{\gamma_{2}}\rightarrow0$.
Therefore, the load-balancing modulation code outperforms the self-randomized
code when $n$ is sufficiently large. 
\end{remrk}

\section{Simulation Results\label{sec:Simulation-Results}}

In this section, we present the simulation results for the modulation
codes described in Sections \ref{sub:Another-rewriting-algorithm}
and \ref{sub:An-enhanced-algorithm}. In the figures, the first modulation
code is called the {}``self-randomized modulation code'' while the
second is called the {}``load-balancing modulation code''. Let the
{}``loss factor'' $\eta$ be the fraction of cell-levels which are
not used when a block erasure is required: $\eta\triangleq1-\frac{E[R]}{n(q-1)}.$
We show the loss factor for random loading with 1 and 2 random choices
as comparison. Note that $\eta$ does not take the amount of information
per cell-level into account. Results in Fig. \ref{Flo:fig2} show
that the self-randomized modulation code has the same $\eta$ with
random loading with 1 random choice and the load-balancing modulation
code has the same $\eta$ with random loading with 2 random choices.
This shows the optimality of these two modulation codes in terms of
ball loading. %
\begin{figure}
\centering{}\includegraphics[scale=0.5]{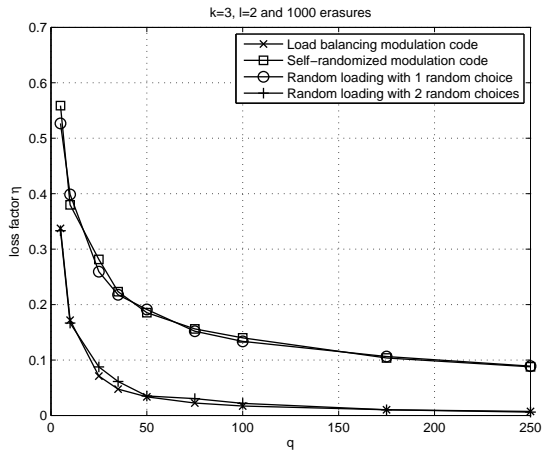}\caption{Simulation results for random loading and algorithms we proposed with
$k=3$, $l=2$ and 1000 erasures.\label{Flo:fig2}}

\end{figure}
\begin{figure}
\centering{}\includegraphics[scale=0.5]{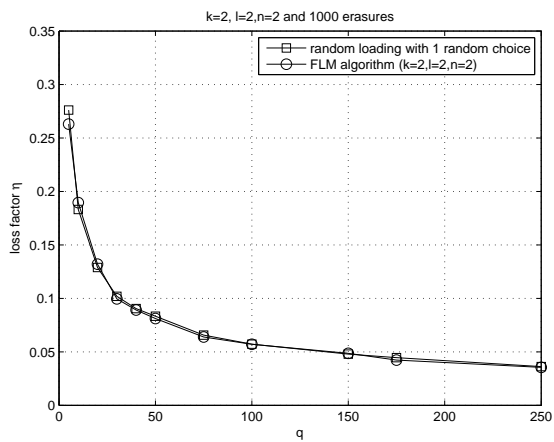}\caption{Simulation results for random loading and codes in {[}4{]} with $k=2$,
$l=2,$ $n=2$ and 1000 erasures.\label{Flo:fig4}}

\end{figure}
\begin{figure}
\centering{}\includegraphics[scale=0.5]{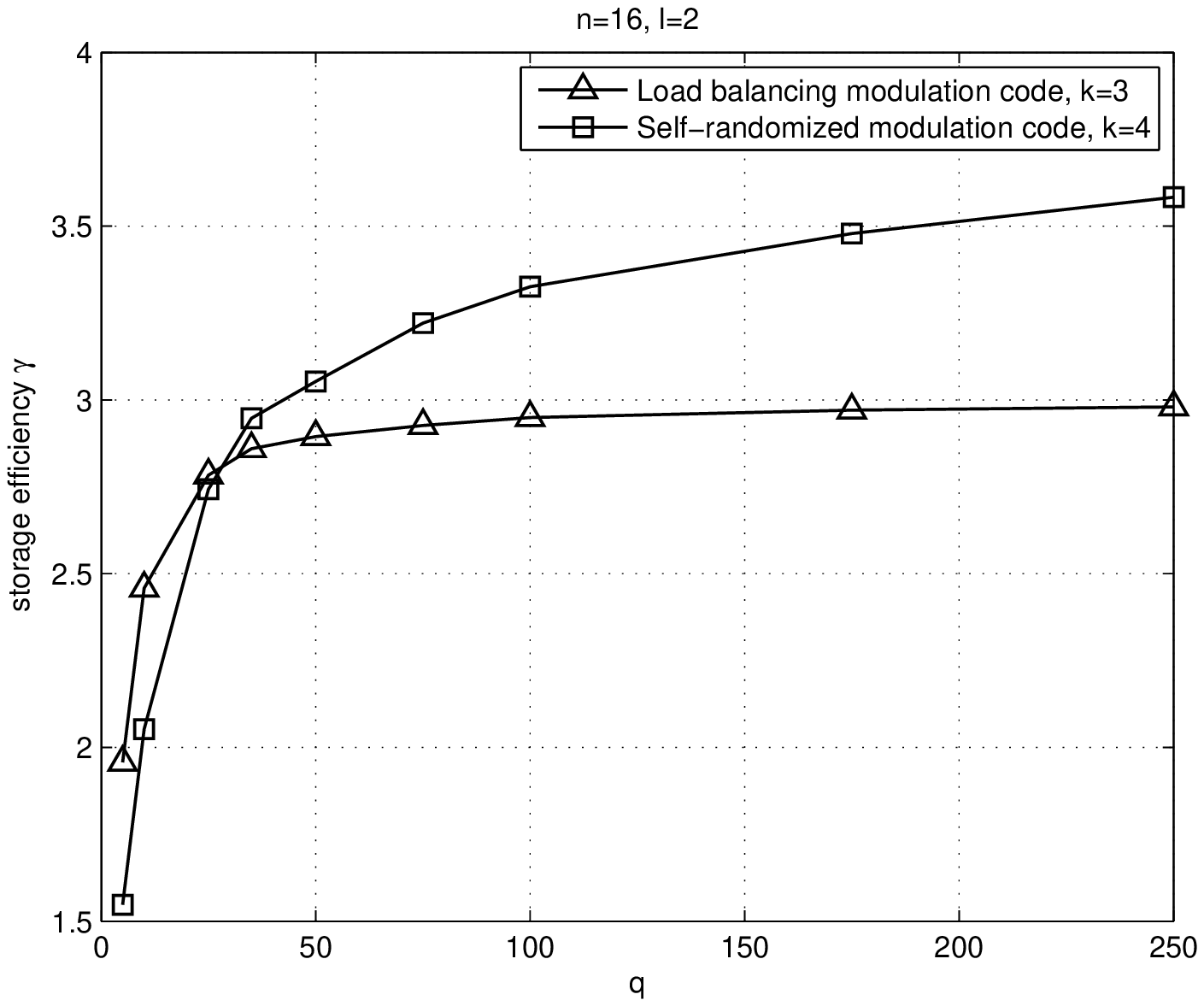}\caption{Storage efficiency of self-randomized modulation code and load-balancing
modulation code with $n=16$.\label{fig:fig5}}

\end{figure}

\begin{figure}

\begin{centering}
\includegraphics[scale=0.5]{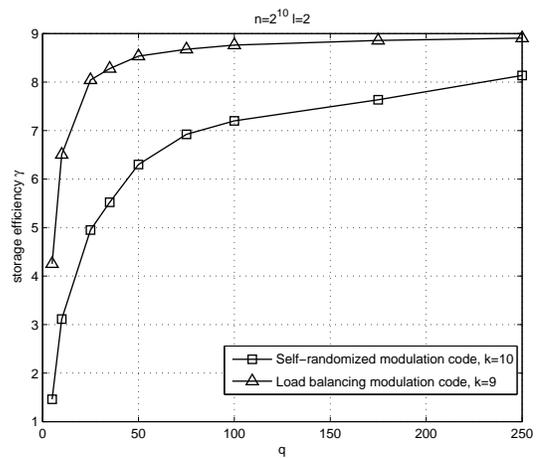}\caption{Storage efficiency of self-randomized modulation code and load-balancing
modulation code with $n=2^{10}$.\label{fig:fig6}}

\par\end{centering}

\end{figure}

We also provide the simulation results for random loading with 1 random
choice and the codes designed in \cite{Finucane_Liu_Mitzenmacher_aller08},
which we denote as FLM-($k=2,l=2,n=2$) algorithm, in Fig. \ref{Flo:fig4}.
From results shown in Fig. \ref{Flo:fig4}, we see that the FLM-($k=2,l=2,n=2$)
algorithm has the same loss factor as random loading with 1 random
choice. This can be actually seen from the proof of asymptotic optimality
in \cite{Finucane_Liu_Mitzenmacher_aller08} as the algorithm transforms
an arbitrary input distribution into an uniform distribution on the
cell-level increment. Note that FLM algorithm is only proved to be
optimal when 1 bit of information is stored. So we just compare the
FLM algorithm with random loading algorithm in this case. Fig. \ref{fig:fig5}
and Fig. \ref{fig:fig6} show the storage efficiency $\gamma$ for
these two modulation codes. Fig. \ref{fig:fig5} and Fig. \ref{fig:fig6}
show that the load-balancing modulation code performs better than
self-randomized modulation code when $n$ is large. This is also shown
by the theoretical analysis in Remark \ref{rem:If--is}.

\section{Conclusion\label{sec:Conclusion}}

\enlargethispage{-5.3in}

In this paper, we consider modulation code design problem for practical
flash memory storage systems. The storage efficiency, or average (over
the distribution of input variables) amount of information per cell-level
is maximized. Under this framework, we show the maximization of the
 number of rewrites for the the worst-case criterion \cite{Ajiang-isit07-01,Ajiang-isit07-02,Yaakobi_verdy_siegel_wolf_allerton08}
and the average-case criterion \cite{Finucane_Liu_Mitzenmacher_aller08}
are two extreme cases of our optimization objective. The self-randomized
modulation code is proposed which is asymptotically optimal for arbitrary
input distribution and arbitrary $k$ and $l$, as the number of cell-levels
$q\rightarrow\infty$. We further consider performance of practical
systems where $q$ is not large enough for asymptotic results to dominate.
Then we analyze the storage efficiency of the self-randomized modulation
code when $q$ is only moderately large. Then the load-balancing modulation
codes are proposed based on the power of two random choices \cite{load_balancing_94}
\cite{Mitzenmacher96thepower}. Analysis and numerical simulations
show that the load-balancing scheme outperforms previously proposed
algorithms.

\end{document}